\documentclass[aps,twocolumn,longbibliography,superscriptaddress]{revtex4-1}
\usepackage{amsmath,amssymb,mathrsfs}
\usepackage{natbib}
\usepackage{subfigure}
\usepackage{tabularx}
\usepackage{epsfig}
\usepackage{longtable}
\usepackage{amsfonts}
\usepackage{rotating}
\usepackage{subfigure}
\usepackage{amsmath}
\usepackage{comment}
\usepackage{bbold} 
\usepackage{multirow}
\usepackage{hhline}
\usepackage{color}
\usepackage{braket}
\usepackage{amsthm}

\def\be{\begin{equation}}
\def\ee{\end{equation}}
\def\bea{\begin{eqnarray}}
\def\eea{\end{eqnarray}}

\newtheorem{theorem}{Theorem}[section]

\newtheorem{lemma}[theorem]{Lemma}
\newtheorem{definition}{Def}
\usepackage{wordlike}

\usepackage[unicode=true,bookmarks=true,bookmarksnumbered=false,bookmarksopen=false,breaklinks=false,pdfborder={0 0 1},backref=false,colorlinks=true]{hyperref}

\hypersetup{linkcolor=magenta,urlcolor=blue,citecolor=blue,pdfstartview={FitH},hyperfootnotes=false,unicode=true}

\begin{document}

\title{Order-preserving condition for coherence measures of projective measurements with One Example } 
\author{Hai Wang}
\affiliation{School of Mathematics and Statistics, Nanjing University of Science and Technology, Nanjing 210094, Jiangsu, China}
\email{quanthw@njust.edu.cn}

\begin{abstract}
   Superposition is an essential feature of quantum mechanics. From the Schr$\ddot{o}$dinger's cat to quantum algorithms such as Deutsch-Jorsza algorithm, quantum superposition plays an important role. It is one fundamental and crucial question how to quantify superposition. Until now, the framework of coherence has been well established as one typical instance of quantum resource theories. And the concept of coherence has been generalized into linearly independent basis, projective measurements and POVMs. In this work, we will focus on coherence measures for projective measurements or orthogonal subspaces. One new condition, order-preserving condition, is proposed for such measures. This condition is rooted in the mathematical structure of Hilbert spaces' orthogonal decomposition. And by generalizing the $\frac{1}{2}$-affinity of coherence into subspace cases, we verify that this generalized coherence measure satisfies the order-preserving condition. And it also satisfies other reasonable conditions to be a good coherence measure. As the partial order relationship exists for not only projective measurements, but also POVMs, it's natural to study the order-preserving condition in POVM cases, which will be the last part of this work.
  
\end{abstract}

\maketitle

Quantum superposition is one fundamental principle in quantum mechanics. From the famous Schr$\ddot{o}$dinger's cat \cite{cat} to double-slit interferences\cite{carnal_young,quach_way} and  quantum applications\cite{zernike_concept,zhang_coherent,DIVINCENZO_1999,yu_total,chitambar_relating,Li_coherence,Pati_17}, quantum superposition plays an crucial role. It is important to quantify superposition properly in either superpositions of states or superpositions of evolution paths. As one typical quantum resource, concepts such as free states and free operations have become standard considerations in quantum resource theories\cite{Gour_rev,wu_imaginary,liu_channel,takagi_general,winter_coherence,Kuroiwa_2020}. And for resource measures, faithfullness, monotonicity, strong monotonicity and convexity have been utilized as one unified framework for quantum resource measures\cite{Gour_rev,Baumgratz_cohere,Vicente_2014,Gallego_15,Streltsov_2018,wu_imaginary}. Until now,  quantum coherence is indispensable for many quantum technologies, such as quantum communications\cite{Duan_00,Halder_2008,Juan_14,Bibak_24}, quantum computations\cite{DIVINCENZO_1999,Jeong_02,Ralph_03,Sha_19}. Although originating form superpositions of orthogonal states, by far, coherence has been generalized into linearly-independent states and POVMs\cite{Aberg_subspace,superposition_coherence,Bischof_povm,Bischof_pra}.

In terms of quantum resource theories, conditions for reasonable coherence measures are firstly proposed in \cite{Baumgratz_cohere}, which are

$(C1)$ Faithfulness. $C(\rho)\geq 0$ and $C(\rho)=0$ if and only if $\rho\in \mathcal{I}$,

$(C2)$ Monotonicity. $C(\Phi(\rho))\leq C(\rho)$ for arbitrary free operation $\Phi$,

$(C3)$ Strong monotonicity. $\sum_i p_i C(\sigma_i)\leq C(\rho)$ with $p_i =Tr{(K_j \rho K_j^{\dagger})}$ and $\sigma_i =p_i^{-1} K_i \rho K_i^{\dagger}$,

$(C4)$ Convexity. $\sum_i p_i C(\rho_i)\geq C(\sum_i p_i \rho_i)$ for arbitrary $\rho_i$ and $p_i\geq 0$ satisfying $\sum_i p_i=1$.

These four conditions have become a paragon for measures about quantum resources, such as entanglement, steering and recent imaginarity and so on\cite{Gour_rev,Baumgratz_cohere,Vicente_2014,Gallego_15,Streltsov_2018,wu_imaginary}. Furthermore, in \cite{Yu_cohere}, they propose an alternative framework for coherence measures, which are proved to be equivalent to criterion in \cite{Baumgratz_cohere}. Concretely, they introduce one new condition, called additivity, to replace conditions $(C3)$ and $(C4)$ in \cite{Baumgratz_cohere}, which is stated as for 
block diagonal states $\rho=p_1 \rho_1\oplus p_2 \rho_2$ in the coherent basis,
$$(C3')\ C(p_1 \rho_1\oplus p_2 \rho_2)=p_1 C(\rho_1)+p_2C(\rho_2).$$ This alternative framework significantly simplifies the justification of coherence measures. 

Although above conditions are about the validity of coherence measures in one fixed orthogonal basis, these ideas boost coherence measures about orthogonal subspaces\cite{Aberg_subspace} and POVMs\cite{Bischof_povm,Bischof_pra}. However, compared with orthogonal basis, there is one remarkable difference for projective subpaces and POVMs. That is, for these two, there exists one partial order relationship, which cannot exist in the traditional coherence setting. And our order-preserving condition is one natural criteria rooted in this difference. Our work will focus on this difference mainly for subspace coherence measures. This work will be devided into three parts. Firstly, we will introduce the order-preserving condition and its motivation in detail. Then, we extend the previous $\frac{1}{2}$-affinity of coherence\cite{Xiong_Pra} into subspace cases, which will be used as an example to show that there are subspace coherence measures satisfying both previous conditions and the order-preserving condition. The last part will be arranged to discuss our order-preserving condition in POVM cases.

\noindent {\it \textbf{Order-Preserving Condition.}---}
In terms of measurements, coherence about orthogonal basis can be regarded as coherence corresponding to von Neumann measurements. For different orthogonal bases, generally we think they are equivalent among each other. Mathematically, this means that there is no partial order relationship in the set of orthogonal bases. However, when we go from von Neumann measurements to projective measurements, things change dramatically. For one Hilbert space $\mathcal{H}$, one projective measurement $\mathbf{P}$ uniquely corresponds to $\mathrm{P}=\{P_m\}$, one orthogonal decomposition of $\mathcal{H}$. On the other hand, for two measurements $\mathbf{P}$ and $\mathbf{Q}$, it is possible that $\mathbf{Q}$ is some coarse graining of $\mathbf{P}$. Namely, every result $q$ case in $\mathbf{Q}$ is one union of some different $p$ cases in $\mathbf{P}$. Combining projective measurements with measurements' coarse graining, we can introduce one partial order relationship for the set of projective measurements. Consider two projective measurements $\mathbf{P}$ and $\mathbf{Q}$, we say $\mathbf{Q}\succeq \mathbf{P}$ if and only if their project projectors $\mathrm{P}=\{P_m\}_{m=0}^{M-1}$ and $\mathrm{Q}=\{Q_n\}_{n=0}^{N-1}$ satisfy that for every $0\leq n\leq N-1$, there is one $0\leq m\leq M-1$ such that for every $0\leq m\leq M-1$, there is one subset $\Lambda_m\subseteq \{0,\ldots,N-1\}$ such that $P_m=\sum_{n\in \Lambda_m}Q_n$. For convenience, projective measurent $\mathbf{P}$ and its projectors $\mathrm{P}=\{P_m\}_{m=0}^{M-1}$ will be regarded as the same.

Intuitively $\mathrm{Q}\succeq \mathrm{P}$ means that $\mathrm{Q}=\{Q_n\}_{n=0}^{N-1}$ is one refinement of $\mathrm{P}=\{P_m\}_{m=0}^{M-1}$. In terms of measurements, for one fixed setting, it is natural to require that if we measure more carefully, then we will not get less. For subspace coherence, given one quantum state $\rho$, suppose that we have one coherence measure $C(\cdot)$. Considering $\mathrm{P}=\{P_m\}_{m=0}^{M-1}$ and $\mathrm{Q}=\{Q_n\}_{n=0}^{N-1}$ with $\mathbf{Q}\succeq \mathbf{P}$, inspired by this fact, it is natural to require that
\begin{equation}
    C_{\mathrm{Q}}(\rho)\geq C_{\mathrm{P}}(\rho).
\end{equation}
We call this property the order-preserving condition.
\begin{definition}[Order-Preserving Condition]
    For two orthogonal decompositions $\mathrm{P}=\{P_m\}_{m=0}^{M-1}$ and $\mathrm{Q}=\{Q_n\}_{n=0}^{N-1}$ of one Hilbert space $\mathcal{H}$, if $\mathrm{Q}\succeq \mathrm{P}$, then $C_{\mathrm{Q}}(\rho)\geq C_{\mathrm{P}}(\rho)$.
\end{definition}

In the next section, we will give one concrete subspace coherence measure, which satisfies both previous conditions and our order-preserving condition.
But before that, we should specify free states and free operations in terms of subspace coherence for completeness.

\begin{definition}
    Given one orthogonal decomposition $\mathrm{P}=\{P_i\}_{i=0}^{M-1}$ of $H$, if $\rho$ belongs to the following set
\begin{equation*}
    \mathcal{I}=\{\rho | \rho=\sum_m P_m \rho P_m\},
    \label{incoherent_subspace}
\end{equation*}then $\rho$ is one free state in terms of $\mathrm{P}$.
\end{definition}

\begin{definition}
    Given one orthogonal decomposition $\mathrm{P}=\{P_i\}_{i=0}^{M-1}$ of $H$, one quantum channel $\Phi$ is free, if it admits one Kraus operator representation $\Phi(\rho)=\sum_i K_i \rho K_i^{\dagger}$ such that for every $i$,
\begin{equation*}
   \frac{K_i \sigma K_i^{\dagger}}{Tr{(K_i \sigma K_i^{\dagger})}}\in \mathcal{I},\forall \sigma\in \mathcal{I}. 
\end{equation*} 
\end{definition}

With these preparations, following the framework in \cite{Yu_cohere}, we say $C(\cdot)$ is one reasonable subspace coherence measure, if it satisfies:

$(C1')$ Faithfulness. $C(\rho)\geq 0$ and $C(\rho)=0$ if and only if $\rho\in \mathcal{I}$.

$(C2')$ Monotonicity. $C(\Phi(\rho))\leq C(\rho)$ for arbitrary free operation $\Phi$.

$(C3')$ Additivity for block-diagonal states. $C(p_1 \rho_1\oplus p_2 \rho_2)=p_1 C(\rho_1)+p_2C(\rho_2).$ 

$(C4')$ Order-preserving. $C_{\mathrm{Q}}(\rho)\geq C_{\mathrm{P}}(\rho)$ for orthogonal decompositions $\mathrm{Q}\succeq \mathrm{P}$.

Conditions $(C1')$-$(C3')$ are direct generalizations of conditions in \cite{Yu_cohere}. And one can also generalize conditions in \cite{Baumgratz_cohere} to subspace coherence cases as done in \cite{Aberg_subspace,Bischof_povm,Bischof_pra}. It is also easy to check that these two generalizations are equivalent using methods in \cite{Yu_cohere}. Although \cite{Aberg_subspace,Bischof_povm,Bischof_pra} generalize coherence into projector and POVM cases,  they didn't notice the mathematical structure of the set of orthogonal decompositions. In terms of subspace coherence measures, our order-preserving condition is one new critia. Now, it's time to give one concrete instance to show that there are subspace coherence measures satisfying $(C1')$-$(C4')$.

\noindent{{\it \textbf{$1/2$-affinity of subspace coherence}---}}
Suppose $\mathcal{H}$ is one $d$-dimensional Hilbert space. Given one orthogonormal basis $\{\ket{\mu_i}\}_{i=0}^{d-1}$, the original $1/2$-affinity of coherence measures the coherence of one state $\rho$ as 
\begin{equation}
     C(\rho)=1-\sum_i \bra{\mu_i}\sqrt{\rho}\ket{\mu_i}^{2},
     \label{original_version}
\end{equation}which was firstly proposed by Xiong and others in \cite{Xiong_Pra}. Using this coherence measure, they reveal nice connections between quantum coherence and path distinguishability\cite{Spehner2014,Spehner2017}.

Now back to our projector setting, consider one orthogonal decomposition $\mathrm{P}=\{P_i\}_{i=0}^{M-1}$ with $M\leq d$. 
\begin{definition}
    Given one orthogonal decomposition $\mathrm{P}=\{P_i\}_{i=0}^{M-1}$ of $\mathcal{H}$, for one state $\rho$, its coherence about $\mathrm{P}$ is 
    \begin{equation}
        C(\rho,\mathrm{P})=1-\sum_{m=0}^{M-1}Tr{[(P_m\sqrt{\rho}P_m)^2]}.
        \label{coherence_sub}
    \end{equation}
\end{definition}
In the following, when the corresponding orthogonal decomposition is unambiguous, we will use $C(\rho)$ directly for convenience.

Before verifying the legitimacy of this measures referring to $(C1')$-$(C4')$, let us see one familiar case. When $\mathrm{P}=\{P_i\}_{i=0}^{M-1}$ corresponds to some orthonormal basis, that is for all $i$, $P_i=\ket{\mu_i}\bra{\mu_i}$, then our $C(\rho)$ will degenerate to the $1/2$-affinity of coherence defined in \cite{Xiong_Pra}, which shows that the quantity in Eq.\eqref{coherence_sub} is one generalization of the measure in Eq.\eqref{original_version} in this sense.
Now, we will verify that $C(\cdot)$ is one reasonable measure for subspace coherence satisfying conditions $(C1')$-$(C4')$.

Firstly, about $C(\cdot)$, we show that it can be re-expressed as the distance between the state $\rho$ and the free state set $\mathcal{I}$.
\begin{lemma}
For states $\rho$, 
   \begin{equation}
     C(\rho)=\min\{d_a^{\frac{1}{2}}(\rho,\sigma)|\sigma=\sum_{m=0}^{M-1}P_m\sigma P_m\in\mathcal{I}\},
     \label{equiv_min}
   \end{equation}where $d_a^{\frac{1}{2}}(\rho,\sigma)=1-[Tr{(\sqrt{\rho}\sqrt{\sigma})}]^2$ is the $\frac{1}{2}$-affinity of distance \cite{Luo_04}. And the minimum can be reached if and only if  
   \begin{equation}
    \sigma_{\rho}=\oplus_m \frac{Tr{[(P_m \sqrt{\rho}P_m)^2]}}{\sum_n Tr{[(P_n \sqrt{\rho}P_n)^2]}}\frac{(P_m \sqrt{\rho}P_m)^2}{Tr{[(P_m \sqrt{\rho}P_m)^2]}}.
    \label{min_condition}
\end{equation}
\end{lemma}

\begin{proof}
    Because $d_a^{\frac{1}{2}}(\rho,\sigma)=1-[Tr{(\sqrt{\rho}\sqrt{\sigma})}]^2$, so the rightside of Eq.\eqref{equiv_min} is equivalent to 
    \begin{equation}
        1-\max_{\sigma\in \mathcal{I}} [Tr{(\sqrt{\rho}\sqrt{\sigma})}]^2,
    \end{equation}where $\mathcal{I}=\{\sigma | \sigma=\sum_{m=0}^{M-1}P_m\sigma P_m\}$.

Note that for $\sigma\in \mathcal{I}$, $\sigma$ must be expressed as 
\begin{equation}
    \sigma=\oplus_m p_m\sigma_m,
    \label{decomposition}
\end{equation}where $\overrightarrow{p}$ is some probability distribution and $\sigma_m$ is some state on the subspace $P_m \mathcal{H}P_m$ for $0\leq m\leq M-1$. Take Eq.\eqref{decomposition} into $Tr{(\sqrt{\rho}\sqrt{\sigma})}$, we have
\begin{equation}
    Tr{(\sqrt{\rho}\sqrt{\sigma})}=\sum_m \sqrt{p_m}Tr{[(P_m \sqrt{\rho}P_m)^{\dagger}\sqrt{\sigma_m}]}.
\end{equation}As $|Tr{(A^{\dagger}B)}|\leq [Tr{(A^{\dagger}A)}]^{\frac{1}{2}}[Tr{(B^{\dagger}B)}]^{\frac{1}{2}}$ with equality established if and only if $B=\alpha A$. For $0\leq m\leq M-1$, let $A$ be $P_m \sqrt{\rho}P_m$ and $B$ be $\sqrt{\sigma_m}$, as $Tr{(\sqrt{\sigma_m}\sqrt{\sigma_m})}=1$, so we know for $0\leq m\leq M-1$,
\begin{equation}
    Tr{[(P_m \sqrt{\rho}P_m)^{\dagger}\sqrt{\sigma_m}]}\leq [Tr{((P_m \sqrt{\rho}P_m)^2)}]^{\frac{1}{2}},
\end{equation} with equality held if and only if $\sqrt{\sigma_m}=\frac{P_m \sqrt{\rho}P_m}{[Tr{((P_m \sqrt{\rho}P_m)^2)}]^{\frac{1}{2}}}$. Now we know that for $ \sigma=\oplus_m p_m\sigma_m$,
\begin{equation}
    Tr{(\sqrt{\rho}\sqrt{\sigma})}\leq \sum_m \sqrt{p_m}[Tr{((P_m \sqrt{\rho}P_m)^2)}]^{\frac{1}{2}},
\end{equation}for the rightside of the above inequality, using the H$\ddot{o}$lder inequality, we have
\begin{equation}
    \sum_m \sqrt{p_m}[Tr{((P_m \sqrt{\rho}P_m)^2)}]^{\frac{1}{2}}\leq [\sum_j Tr{((P_j \sqrt{\rho}P_j)^2)}]^{\frac{1}{2}},
\end{equation}with equality held if and only if $p_m=\frac{Tr{[(P_m \sqrt{\rho}P_m)^2]}}{\sum_n Tr{[(P_n \sqrt{\rho}P_n)^2]}}$, for $0\leq m\leq M-1$.

Above all, we know that for $\sigma\in \mathcal{I}$,
\begin{equation}
    Tr{(\sqrt{\rho}\sqrt{\sigma})}\leq [\sum_m Tr{[(P_m \sqrt{\rho}P_m)^2]}]^{\frac{1}{2}},
\end{equation}which turns into one equation if and only if 
\begin{equation*}
    \sigma=\oplus_m \frac{Tr{[(P_m \sqrt{\rho}P_m)^2]}}{\sum_n Tr{[(P_n \sqrt{\rho}P_n)^2]}}\frac{(P_m \sqrt{\rho}P_m)^2}{Tr{[(P_m \sqrt{\rho}P_m)^2]}}.
\end{equation*} That is, $\max_{\sigma\in \mathcal{I}} [Tr{(\sqrt{\rho}\sqrt{\sigma})}]^2=\sum_m Tr{[(P_m \sqrt{\rho}P_m)^2]}$. On the other hand, by definition, we have $C(\rho)=1-\sum_{m=0}^{M-1}Tr{[(P_m\sqrt{\rho}P_m)^2]}$. So
\begin{equation*}
    C(\rho)=1-\max_{\sigma\in \mathcal{I}} [Tr{(\sqrt{\rho}\sqrt{\sigma})}]^2=\min_{\sigma\in \mathcal{I}} d_a^{\frac{1}{2}}(\rho,\sigma).
\end{equation*}
\end{proof}
With this Lemma, it's obvious that $C(\rho)=0$ if and only if $\rho\in \mathcal{I}$, which shows our $C(\cdot)$ satisfies $(C1')$ condition. What's more, as $d_a ^{\frac{1}{2}}(\cdot ,\cdot)$ is contractive under quantum operations, the above Lemma guarantees that $C(\Phi(\rho))\leq C(\rho)$ for any incoherent quantum operation $\Phi$ and quantum state $\rho$. Thus,our $C(\cdot)$ satisfies $(C2')$ condition as well. In particular, for arbitrary $U$ that can be decomposed as $ U=\bigoplus_m U_m$, where each $U_m$ is one arbitrary unitary operator on the subspace corresponding to $P_m$, we have
$C(U\rho U^{\dagger})=C(\rho)$, which is also one requirement in \cite{Aberg_subspace}.

Next, we will show that our $C(\Phi(\rho))$ satisfies $C(3')$, the additivity condition.
\begin{lemma}
$C_{\frac{1}{2}}(\rho)$ satisfies the additivity condition in the following sense, for any decomposition $\mathcal{H}_1=\bigoplus_m P_m$, if $\rho=p\rho_1\oplus (1-p)\rho_2$ is block diagonal in respect to $\bigoplus_{m=0}^{M-1} P_m$, then 
    \begin{equation}
       C(\rho)=C(p\rho_1\oplus (1-p)\rho_2)=p C(\rho_1)+(1-p)C(\rho_2).
    \end{equation} 
\end{lemma}    
\begin{proof}
As $\rho=p\rho_1\oplus (1-p)\rho_2$ is block diagonal in respect to $\bigoplus_m P_m$, so there must two disjoint set $\Lambda_1$ and $\Lambda_2$ of $\{0,\ldots,M-1\}$ such that 
\begin{equation*}
    \rho_1\in \bigoplus_{m\in\Lambda_1}P_m, \ \rho_2\in \bigoplus_{m\in\Lambda_2}P_m.
\end{equation*}Then the following will be true:   

\begin{align*}                                                 
        &C(\rho)=C(p\rho_1\oplus (1-p)\rho_2)\\
        &=1-\sum_{m\in \Lambda_1} Tr\{[(P_m\oplus0)\sqrt{p\rho_1\oplus (1-p)\rho_2}(P_m\oplus0)]^2\}-\\
        &\sum_{n\in\Lambda_2} Tr\{[(0\oplus P_n)\sqrt{p\rho_1\oplus (1-p)\rho_2}(0\oplus P_n)]^2\}\\
        &=1-p\sum_{m\in\Lambda_1}Tr[(P_m\sqrt{\rho_1}P_m)^2]-(1-p)\sum_{n\in\Lambda_2}Tr[(P_n\sqrt{\rho_2}P_n)^2]\\
        &=pC(\rho_1)+(1-p) C(\rho_2).       
\end{align*}  
Thus, condition $(C4')$ is verified.
\end{proof}

Now, it's high time to show that our $C(\cdot)$ satisfies the order preserving condition
\begin{lemma} 
$C(\rho)$ satisfies the order preserving condition. That is, for two orthogonal decompositions $\mathrm{P}=\{P_m\}_{m=0}^{M-1}$ and $\mathrm{Q}=\{Q_n\}_{n=0}^{N-1}$ of $\mathcal{H}$, if $\mathrm{Q}\succeq \mathrm{P}$, then $C(\rho,\mathrm{Q})\geq C(\rho,\mathrm{P})$.
\end{lemma}
\begin{proof}
    Firstly, by definition, for these two orthogonal decompositions $\mathrm{P}=\{P_m\}_{m=0}^{M-1}$ and $\mathrm{Q}=\{Q_n\}_{n=0}^{N-1}$, we have  $C_{\mathrm{P}}(\rho)=1-\sum_{m=0}^{M-1}Tr{[(P_m\sqrt{\rho}P_m)^2]}$ and $C_{\mathrm{Q}}(\rho)=1-\sum_{n=0}^{N-1}Tr{[(Q_n\sqrt{\rho}Q_n)^2]}$.
    
    Because $\mathrm{Q}\succeq \mathrm{P}$, without losing generality, we assume that $P_0=Q_0+Q_1$. Then for $Tr{[(P_0\sqrt{\rho}P_0)^2]}$, it equals to 
    \begin{equation*}
        \sum_{i=0}^1 Tr{[(Q_i\sqrt{\rho}Q_i)^2]}+\sum_{0\leq i\neq j\leq 1}Tr{[(Q_i \sqrt{\rho}Q_j)^{\dagger}(Q_i \sqrt{\rho}Q_j)]},
    \end{equation*}which results in
    \begin{equation}
        Tr{[(P_0\sqrt{\rho}P_0)^2]}\geq \sum_{i=0}^1 Tr{[(Q_i\sqrt{\rho}Q_i)^2]}
        \label{bigger}
    \end{equation} for arbitrary $\rho$.
    Because $\mathrm{Q}\succeq \mathrm{P}$, so for every $0\leq m\leq M-1$, there is one subset $\Lambda_m\subseteq \{0,\ldots,N-1\}$ such that $P_m=\sum_{n\in \Lambda_m}Q_n$. Simliar as Eq.\eqref{bigger}, we have
    \begin{equation}
        Tr{[(P_m\sqrt{\rho}P_m)^2]}\geq \sum_{n\in \Lambda_m} Tr{[(Q_n\sqrt{\rho}Q_n)^2]},
    \end{equation}which means that
    \begin{align*}
         \sum_{m=0}^{M-1}Tr{[(P_m\sqrt{\rho}P_m)^2]}&\geq \sum_{m=0}^{M-1}\sum_{n\in \Lambda_m} Tr{[(Q_n\sqrt{\rho}Q_n)^2]}\\
         &=\sum_{n=0}^{N-1}Tr{[(Q_n\sqrt{\rho}Q_n)^2]}.
    \end{align*}
   So, for two orthogonal decompositions $\mathrm{P}=\{P_m\}_{m=0}^{M-1}$ and $\mathrm{Q}=\{Q_n\}_{n=0}^{N-1}$ of one Hilbert space, if $\mathrm{Q}\succeq \mathrm{P}$, then we have $C(\rho,\mathrm{Q})\geq C(\rho,\mathrm{P})$.
\end{proof}

Until now, we can say that
\begin{theorem}
    Given an orthogonal decomposition $\mathrm{P}=\{P_i\}_i$ of the Hilbert space $\mathcal{H}$, the following quantity
    \begin{equation*}
        C(\rho,\mathrm{P})=1-\sum_{m=0}^{M-1}Tr{[(P_m\sqrt{\rho}P_m)^2]}
    \end{equation*}satisfies conditions $(C1')$-$(C4')$, which is one reasonable block coherence measure.
\end{theorem}
 Besides, given one orthogonal decomposition $\mathrm{P}=\{P_m\}_{m=0}^{M-1}$, we also show maximally coherent states defined by $C(\cdot)$ in the following lemma. 

\begin{lemma}
    Given one Hilbert space $\mathcal{H}$ and its one orthogonal decomposition $\mathrm{P}=\{P_m\}_{m=0}^{M-1}$, the sufficient and nesscessary condtion for $\ket{\psi}$ to be maximally coherent is $\ket{\psi}=\frac{1}{\sqrt{M}}\sum_m \ket{\psi_m}$, with $\ket{\psi_m}\in P_m(\mathcal{H}), \forall 0\leq m\leq M-1$. 
\end{lemma}

\begin{proof}
    Given one pure state $\ket{\psi}$, from Eq.\eqref{coherence_sub}, its coherence about $\mathrm{P}=\{P_m\}_{m=0}^{M-1}$ is
    \begin{equation}
        C_{\frac{1}{2}}(\ket{\psi})=1-\sum_{m=0}^{M-1}||P_m\ket{\psi}||^4.
    \end{equation}So the statement that $\ket{\psi}$ is maximally coherent is equivalent to $\sum_{m=0}^{M-1}||P_m\ket{\psi}||^4$ being minimal. Note that for $\{||P_m\ket{\psi}||\}_m$, we have $\sum_{m=0}^{M-1}||P_m\ket{\psi}||^2=1.$ Based on these, we know that $\sum_{m=0}^{M-1}||P_m\ket{\psi}||^4$ being minimal if and only if  $||P_m\ket{\psi}||=\frac{1}{\sqrt{M}}, \forall 0\leq m\leq M-1$.
\end{proof}

\noindent {\it \textbf{Order-Preserving in POVM Cases.}---}
Order-preserving condition is the mathematical formulation of measurements' coarse-graining. On the other hand, we know that coarse-graining is quite general for measurements, which is not restricted to projective measurements. Therefore, it's worth discussing the order-preserving condition in the realm of the most general measurements, POVMs. 

The resource theory of coherence based on POVMs is firstly proposed in \cite{Bischof_povm}. In \cite{Bischof_povm}, they define POVM-based coherence by the famous Naimark extension, where corresponding free states and free operations are also defined. And in \cite{Bischof_povm}, given one POVM $\mathrm{E}=\{E_i\}_i$, they give one analytical form of a POVM-based coherence measure, which is expressed as
\begin{equation}
    C_{rel}(\rho)=H[\{p_i(\rho)\}]+\sum_i p_i(\rho)S(\rho_i)-S(\rho)
    \label{rel_POVM}.
\end{equation}In \eqref{rel_POVM}, $p_i(\rho)=Tr(E_i \rho)$ and $\rho_i=(1/p_i)A_i\rho_i A^{\dagger}$, where $\mathrm{A}=\{A_i\}_i$ is one measurement operators corresponding to $\mathrm{E}=\{E_i\}_i$. They verify that this measure does not depend on the choice of measurement operators and satisfies the criterion to be a good POVM coherence measure. Although in \cite{Bischof_povm}, they mainly consider faithfulness, monotonicity and convexity conditions similiar to $(C1)$, $(C2)$ and $(C4)$ mentioned in the introduction part, in the following, we will show that the POVM coherence measure \eqref{rel_POVM} also satisfies our order-preserving condition.

Let us firstly recall what coarse-graining means in the POVM setting. Given one POVM $\mathrm{E}=\{E_i\}_{i=0}^{M-1}$ with its measurement operators $\mathrm{A}=\{A_i\}_{i=0}^{M-1}$, we know that the indice $i$ represents the measurement outcome. Now suppose that there is another POVM $\mathrm{F}=\{F_j\}_{j=0}^{N-1}$. If $\mathrm{F}$ is one coarse-graining of $\mathrm{E}$, similar as the projective measurement case, this means that every outcome $j$ of $\mathrm{F}$ is one union of some outcomes of $\mathrm{E}$. That is, for every $j$ in $\{0,\ldots,N-1\}$, there is one subset $\Lambda_j$ of $\{0,\ldots,M-1\}$ such that 
\begin{equation}
    q_j(\rho)=Tr(F_j\rho)=\sum_{i\in \Lambda_j}Tr(E_i\rho)=\sum_{i\in \Lambda_j}p_i(\rho),
    \label{povm_coarse_prob}
\end{equation} and $\mathrm{F}$'s $j$ outcome state can be expressed as 
\begin{equation}
    \rho_j=\sum_{i\in \Lambda_j}\frac{p_i(\rho)}{\sum_{i\in \Lambda_j}p_i(\rho)}\cdot \frac{A_i\rho A_i^{\dagger}}{p_i(\rho)}=\sum_{i\in \Lambda_j}\frac{p_i(\rho)}{\sum_{i\in \Lambda_j}p_i(\rho)}\rho_i.
    \label{povm_coarse_s}
\end{equation}Besides, $\{\Lambda_j\}_{j=0}^{N-1}$ forms one disjoint union of $\{0,\ldots,M-1\}$.

Now, we will show that the POVM coherence measure \eqref{rel_POVM} satisfies the order-preserving condition.
\begin{theorem}[Order-Preserving For POVMs]
Given two POVM $\mathrm{E}=\{E_i\}_{i=0}^{M-1}$ and $\mathrm{F}=\{F_j\}_{j=0}^{N-1}$. If $\mathrm{E}\succeq\mathrm{F}$, that is, $\mathrm{F}$ is one coarse-graining of $\mathrm{E}$, then, for every state $\rho$, 
\begin{equation*}
    C_{rel}(\rho,\mathrm{E})\geq C_{rel}(\rho,\mathrm{F}).
\end{equation*}
\end{theorem}
\begin{proof}
Suppose that $\mathrm{A}=\{A_i\}_i$ is one measurement operators corresponding to $\mathrm{E}=\{E_i\}_i$, then we know that
$p_i\doteq p_i(\rho)=Tr(E_i\rho)$ and its $i$ outcome state is $\rho_i=(1/p_i)A_i\rho_i A^{\dagger}$.

Because $\mathrm{E}\succeq\mathrm{F}$, so utilizing the definition of $C_{rel}(\cdot)$, Eq.\eqref{povm_coarse_prob} and Eq.\eqref{povm_coarse_s}, to prove 
\begin{equation*}
    C_{rel}(\rho,\mathrm{E})\geq C_{rel}(\rho,\mathrm{F})
\end{equation*} is equivalent to prove
\begin{equation}
    H[\{p_i\}_{i=0}^{M-1}]+\sum_{i=0}^{M-1}p_iS(\rho_i)\geq  H[\{q_j\}_{j=0}^{N-1}]+\sum_{j=0}^{N-1}q_jS(\rho_j).
    \label{equivalent}
\end{equation}
Note that for $\sum_{j=0}^{N-1}q_jS(\rho_j)$, because of $\rho_j=\sum_{i\in \Lambda_j}\frac{p_i(\rho)}{\sum_{i\in \Lambda_j}p_i(\rho)}\rho_i$, so for every $j\in \{0,\ldots,N-1\}$, we have
\begin{equation*}
    S(\rho_j)=S(\sum_{i\in \Lambda_j}\frac{p_i(\rho)}{\sum_{i\in \Lambda_j}p_i(\rho)}\rho_i)\leq \sum_{i\in\Lambda_j}\frac{p_i}{q_j}S(\rho_i)+H[\{\frac{p_i}{q_j}\}_{i\in\Lambda_j}],
\end{equation*}which results in
\begin{equation}
    \sum_{j=0}^{N-1}q_jS(\rho_j)\leq \sum_{i=0}^{M-1}p_i S(\rho_i)+\sum_{j=0}^{N-1}q_j H[\{\frac{p_i}{q_j}\}_{i\in\Lambda_j}].
    \label{mid_step}
\end{equation}

Now taking Eq.\eqref{mid_step} into the rightside of Eq.\eqref{equivalent}, to establish Eq.\eqref{equivalent}, it is equivalent to prove the following fact
\begin{equation}
    H[\{p_i\}_{i=0}^{M-1}]\geq \sum_{j=0}^{N-1}(\sum_{i\in\Lambda_j}p_i)H[\{\frac{p_i}{q_j}\}_{i\in\Lambda_j}].
    \label{equivalent_2}
\end{equation}By direct computations, we know that
\begin{align*}
    & \sum_{j=0}^{N-1}(\sum_{i\in\Lambda_j}p_i)H[\{\frac{p_i}{q_j}\}_{i\in\Lambda_j}]\\
    &=\sum_{j=0}^{N-1}(\sum_{i\in\Lambda_j}p_i)(-\sum_{i\in\Lambda_j}\frac{p_i}{\sum_{i\in\Lambda_j}p_i}\log \frac{p_i}{\sum_{i\in\Lambda_j}p_i})\\
    &=-\sum_{i=0}^{M-1}p_i \log p_i +\sum_{j=0}^{N-1}(\sum_{i\in \Lambda_j}p_i)\log (\sum_{i\in \Lambda_j}p_i)\\
    &= H[\{p_i\}_{i=0}^{M-1}]- H[\{q_j\}_{j=0}^{N-1}]\leq H[\{p_i\}_{i=0}^{M-1}].
\end{align*}Thus, Eq.\eqref{equivalent_2} is correct, showing the correctness of Eq.\eqref{equivalent}. Until now, $C_{rel}(\rho,\mathrm{E})\geq C_{rel}(\rho,\mathrm{F})$ is verified.
\end{proof}

\noindent {\it \textbf{Discussion and Conclusion.}---}
In this work, we propose the order-preserving condition as one new criterion for block coherence. This condition is the mathematical formulation of the coarse-graining of measurements. Given one state, for block coherence, this condition bridges the finest measurement, von-Neumann measurements and the crudest measurement $\mathrm{E}=\{I\}$. Once we observe more carefully, the corresponding amount of coherence will not decrease. Combined the order-preserving condition with previous block coherence criteria, we show that there are coherence measures satisfying all these requirements. And by extending the $1/2$-affinity of coherence into block coherence cases, we obtain one concrete block coherence measure which satisfies all requirements. Furthermore, in terms of the generality of POVMs in quantum mechanics, we also study the order-preserving condition in the POVM coherence case. Based on \cite{Bischof_povm,Bischof_pra}, we verify that their POVM coherence measure $C_{rel}(\cdot)$ also satisfies our order-preserving condition.

Further attempts are still needed. For example, in the POVM coherence part, we checked wether one existed coherence measure satisfies the order-preserving condition. Fortunately, it satisfies the order-preserving condition. However, previous considerations about coherence measures are mainly focused on one measurement cases, rarely about comparisons between different measurements. Faithfulness, monotonicity and convexity are their main concerns. About the fact $C_{rel}(\cdot)$ satisfying our order-preserving condition, is it just a coincidence ? Or can we derive the order-preserving condition from those existed conditions? Given one POVM $\mathrm{E}=\{E_i\}_i$, considering our block coherence measure given in Eq.\eqref{coherence_sub}, there is one natural extension into POVM cases,
\begin{equation}
    C(\rho)=1-\sum_i Tr[(A_i\sqrt{\rho}A_i^{\dagger})^2],
    \label{coherence_povm}
\end{equation}where $\mathrm{A}=\{A_i\}_i$ is one measurement operator set of $\mathrm{E}$. Can this quantity be one proper POVM coherence measure? By direct computations, the quantity in Eq.\eqref{coherence_povm} does not depend on the choice of measurement operators. Secondly, in \cite{Bischof_povm}, they give one nice description of free states in the framework of POVM coherence. If we assume that corresponding sets of free states are not empty, then based on their result, it's easy to verify that the quantity given in Eq.\eqref{coherence_povm} satisfies condtions $(C1')$-$(C3')$. And lastly, this quantity also satisfies our order-preserving condition
\begin{lemma}
    Given two POVMs $\mathrm{E}=\{E_i\}_{i=0}^{M-1}$ and $\mathrm{F}=\{F_j\}_{j=0}^{N-1}$, if $\mathrm{F}$ is one coarse-graining of $\mathrm{E}$, then for arbitrary state $\rho$, we have
    \begin{equation*}
        C(\rho,\mathrm{E})\geq C(\rho,\mathrm{F}).
    \end{equation*}
\end{lemma}
\begin{proof}
    Because $\mathrm{F}$ is one coarse-graining of $\mathrm{E}$, so for every $j\in\{0,\ldots,N-1\}$, there will be a subset $\Lambda_j\subseteq\{0,\ldots,M-1\}$ such that
    \begin{equation}
        F_j=\sum_{i\in\Lambda_j}E_i
        \label{coarse}
    \end{equation} and $\cup_{j}\Lambda_j=\{0,\ldots,M-1\}$.

    To prove this lemma, it is equivalent to show that 
    \begin{equation}
        \sum_i Tr[(E^{1/2}_i\sqrt{\rho}E^{1/2}_i)^2]\leq  \sum_j Tr[(F^{1/2}_j\sqrt{\rho}F^{1/2}_j)^2].
    \end{equation}Taking Eq.\eqref{coarse} into consideration, about $Tr[(F^{1/2}_j\sqrt{\rho}F^{1/2}_j)^2],\forall j$, we have
    \begin{align*}
       & Tr[(F^{1/2}_j\sqrt{\rho}F^{1/2}_j)^2]=Tr\{[(\sum_{m\in\Lambda_j}E_m)^{1/2}\sqrt{\rho}(\sum_{n\in\Lambda_j}E_n)^{1/2}]^2\}\\
       &=Tr[(\sum_{m\in\Lambda_j}E_m)\sqrt{\rho}(\sum_{n\in\Lambda_j}E_n)\sqrt{\rho}]\\
       &=\sum_{m\in\Lambda_j}Tr(E_m\sqrt{\rho}E_m\sqrt{\rho})+\sum_{m\neq n, m,n\in\Lambda_j}Tr[(E^{1/2}_m\sqrt{\rho}E^{1/2}_n)(E^{1/2}_m\sqrt{\rho}E^{1/2}_n)^{\dagger}]\\
       &\geq \sum_{m\in\Lambda_j}Tr(E_m\sqrt{\rho}E_m\sqrt{\rho})=\sum_{m\in\Lambda_j}Tr[(E^{1/2}_m\sqrt{\rho}E^{1/2}_m)^2].
    \end{align*}Togather with the fact that $\cup_{j}\Lambda_j=\{0,\ldots,M-1\}$, finally we get the following fact
    \begin{equation}
        \sum_i Tr[(E^{1/2}_i\sqrt{\rho}E^{1/2}_i)^2]\leq  \sum_j Tr[(F^{1/2}_j\sqrt{\rho}F^{1/2}_j)^2],
    \end{equation} which is just what we want.
\end{proof}
Until now, it seems fine that Eq.\eqref{coherence_povm} is one proper POVM coherence measure. But as what is revealed in \cite{Bischof_povm}, for many POVMs, their corresponding coherence-free state set can be empty, for which previous considerations won't work. So to completely verify the quantity is one proper POVM coherence measure, further attempts are needed.

\noindent {\it \textbf{Acknowledgement.}}
%\paragraph*{\textbf{Acknowledgement.}} 
This work is supported by the National
Natural Science Foundation of China, Grant No. 1240010163 and the Fundamental Research Funds for the Central Universities, No.30925010422.

\bibliography{references}

\newpage

\end{document}